\documentclass[letterpaper, 10 pt, twocolumn]{article}

\usepackage{cite}
\usepackage{amsmath,amssymb,amsfonts,amsthm}
\usepackage{graphicx}

\usepackage{mathtools}
\usepackage{balance}
\usepackage{comment}
\usepackage[dvipsnames]{xcolor}
\usepackage{cancel}
\usepackage{xfrac}
\usepackage{soul} 
\usepackage[ruled,vlined,linesnumbered]{algorithm2e} 
\usepackage{subcaption}  
\DontPrintSemicolon
\usepackage{booktabs} 

\definecolor{BlueMATLAB}{rgb}{0, 0.447, 0.741}
\definecolor{RedMATLAB}{rgb}{0.85, 0.325, 0.098}
\definecolor{YellowMATLAB}{rgb}{0.929, 0.694, 0.125}
\definecolor{GreenMATLAB}{rgb}{0.466, 0.674, 0.188}

\newcommand{\myldots}{\kern-0.05em.\kern-0.01em.\kern-0.01em.\kern0.01em}
\newcommand{\R}{\mathbb{R}}
\newcommand{\I}{\mathbb{I}}

\newcommand{\norm}[1]{\left\lVert#1\right\rVert}
\newcommand{\fe}{\mathsf{f}}
\newcommand{\ve}{\mathrm{v}}

\newcommand{\me}{\mathrm{m}}

\newcommand\scalemath[2]{\scalebox{#1}{\mbox{\ensuremath{\displaystyle #2}}}}

\newtheorem{corollary}{Corollary}
\newtheorem{remark}{Remark}

\newtheorem{proposition}{Proposition}
\newtheorem{assumption}{Assumption}

\usepackage[final]{hyperref} 
\hypersetup{
colorlinks=true, 
linkcolor=black, 
citecolor=black, 
filecolor=magenta, 
urlcolor=black         
}

\newcommand{\email}[1]{\href{mailto:#1}{#1}}

\title{\LARGE \bf Efficient Configuration-Constrained Tube MPC via Variables Restriction and Template Selection}

\author{Filippo Badalamenti, Sampath Kumar Mulagaleti, Mario Eduardo Villanueva,\\Boris Houska, Alberto Bemporad\thanks{F.~Badalamenti (Corresponding Author, \email{filippo.badalamenti@imtlucca.it}), S.K.~Mulagaleti, M.E.~Villanueva and A.~Bemporad are with IMT School for Advanced Studies Lucca, Piazza San Francesco 19, 55100, Italy. B.~Houska is with ShanghaiTech University, China. This work has received support from the European Research Council (ERC), Advanced Research Grant COMPACT (Grant Agreement No. 101141351)}}

\date{}

\begin{document}

\maketitle
\thispagestyle{empty}
\pagestyle{empty}

\begin{abstract}
Configuration-Constrained Tube Model Predictive Control (CCTMPC) offers flexibility by using a polytopic parameterization of invariant sets and the optimization of an associated vertex control law. This flexibility, however, often demands computational trade-offs between set parameterization accuracy and optimization complexity. This paper proposes two innovations that help the user tackle this trade-off. First, a structured framework is proposed, which strategically limits optimization degrees of freedom, significantly reducing online computation time while retaining stability guarantees. This framework aligns with Homothetic Tube MPC (HTMPC) under maximal constraints. Second, a template refinement algorithm that iteratively solves quadratic programs is introduced to balance polytope complexity and conservatism. Simulation studies on an illustrative benchmark problem as well as a high-dimensional ten-state system demonstrate the approach’s efficiency, achieving robust performance with minimal computational overhead. The results validate a practical pathway to leveraging CCTMPC’s adaptability without sacrificing real-time viability.
\end{abstract}

\section{Introduction}
Tube Model Predictive Control (TMPC) has proven to be a central tool in robust control due to its ability to enclose all realizations of the state of an uncertain system within a sequence of bounded sets~\cite{Houska2019,Kouvaritakis2015,Rawlings2015}. Among the various set parameterizations, a variety of polytopic set parameterization approaches are particularly attractive, because they permit the online solution of a Quadratic Programming (QP) problem~\cite{houska2024polyhedralcontroldesigntheory}. Examples for such polytopic schemes include Rigid~\cite{MAYNE_RakovicRigidTMPC2005219,Implicit_rakovic2023}, Homothetic~\cite{HTMPC_RAKOVIC20121631}, and Elastic~\cite{ETMPC_7525471, Fleming2015} Tube MPC.
While these methods typically enforce set invariance through linear feedback laws, previous works have also explored the use of a freely parameterized control law within the Homothetic framework~\cite{Langson2004}.

Departing from conventional parameterization methods, Configuration-Constrained Tube MPC (CCTMPC)~\cite{CCTMPC_VILLANUEVA2024111543} introduces a novel framework centered on selecting polytopic sets with fixed facet orientations that share a predefined geometric structure. This approach optimizes facet displacements while explicitly enumerating the vertices of the polytopes and their corresponding control inputs, inherently supporting systems with multiplicative uncertainty.

Despite its merits, CCTMPC presents dual challenges. First, its formulation inherently involves a larger number of optimization variables—encompassing both control inputs and set geometry—compared to methods that predefine either component. This complexity inevitably grows with system dimensionality. Second, the control scheme’s efficacy hinges critically on the template design, particularly in high-dimensional spaces. Here, achieving a balance between template expressiveness and computational tractability becomes nontrivial, as overly simplistic templates risk conservatism, while intricate ones incur prohibitive online costs.

\textit{Contribution:} This paper advances the state-of-the-art in CCTMPC through two novel methods that enable computational efficiency and template design flexibility.

Here, we first introduce a novel family of Tube MPC controllers that systematically reduce computational overhead. By strategically constraining the number of optimization variables, our approach transitions from the full flexibility of CCTMPC to a highly efficient scheme reminiscent of Homothetic Tube MPC (HTMPC). Moreover, we propose an iterative algorithm for designing configuration templates. This algorithm is initialized with a low-complexity template. It progressively enriches the template by solving QPs, systematically increasing the complexity until a user-defined target is achieved.

\textit{Outline}:  
The paper is structured as follows:
\begin{itemize}
\item Section~\ref{sec:ProbForm} formalizes the problem and establishes the necessary mathematical framework.
\item Section~\ref{sec:TubeSchemes} derives the proposed family of reduced-complexity Tube MPC schemes, discussing their theoretical properties.
\item Section~\ref{sec:TemplateCompute} details the adaptive template refinement procedures and NLP formulation.
\item Section~\ref{sec:NumExamples} validates the approach through numerical simulations, including numerical results for a 10-state quadcopter model.
\item Section~\ref{sec:Conclusions} summarizes the contributions and outlines directions for future research.
\end{itemize}

\textit{Notation}: The symbol $\I_n$ denotes the identity matrix of order $n$. A block-diagonal matrix composed of matrices $Q_{1},\dots,Q_{n}$ is denoted by $\mathrm{blkd}(Q_{1},\dots,Q_{n})$. The operator $\operatorname{convh}$ denotes the convex hull, and $\norm{z}_Q^2 \coloneqq z^\top Qz$. Given compact sets $\mathcal{A} \subseteq \mathcal{B} \subset \R^n$, the operator $\mathrm{d}_{\mathcal{B}}(\mathcal{A})$ measures the $2$-norm Hausdorff distance between the sets.

\section{Tube Model Predictive Control} \label{sec:ProbForm}
This section briefly reviews the main idea of Tube MPC and introduces the notation used throughout this article.

\subsection{Uncertain Linear Systems}
This paper is concerned with uncertain discrete-time systems of the form
\begin{align}
\label{eq:system}
x_{k+1} = A_k x_k + B_k u_k + w_k,
\end{align}
where \(x_k \in \mathcal{X} \subseteq \mathbb{R}^{n_x}\) and \(u_k \in \mathcal{U} \subseteq \mathbb{R}^{n_u}\) denote the state and input at time \(k \in \mathbb{N}\), while \(\mathcal{X}\) and \(\mathcal{U}\) are given, closed and convex, state and control constraint sets. Moreover, \(w_k \in \mathcal{W}\) is an additive disturbance, with \(\mathcal{W}\) assumed to be compact and convex. The matrices \(A_k\) and \(B_k\) are uncertain, but known to satisfy 
\[
(A_k,B_k) \in \Delta\coloneqq\operatorname{convh}( (A_1,B_1),\myldots,(A_\me,B_\me) ).
\]
Here, the matrices \((A_i,B_i)\) are assumed to be given vertices of the matrix polytope \(\Delta\) that models the uncertainty set.

\subsection{Robust Control Tubes}
Let \(\mathcal{F}(X)\) denote the set of sets that can be reached from a given set \(X \subseteq \R^{n_x}\); that is,
\begin{align*}
\mathcal{F}(X) := \left\{ X^+ \subseteq \R^{n_x} \ \middle| \ 
\begin{aligned} 
&\forall \, x \in X, \ \exists \, u \in \mathcal{U} : \vspace{1pt} \\
&\forall \, (A,B) \in \Delta, \; w \in \mathcal{W}, \vspace{1pt}  \\
&Ax+Bu+w \in X^+
\end{aligned} 
\right\}.
\end{align*}
We recall that a sequence of sets
\(X_0, X_1,\ldots \subseteq \mathcal X\) is called a feasible Robust Control Tube (RCT) if
\begin{align}
\label{eq::RCT}
\forall k \in \{0,\dots,N-1\},\quad X_{k+1} \in \mathcal{F}(X_k).
\end{align}
Moreover, a set \(X \subseteq \mathcal X\) is called a feasible Robust Control Invariant (RCI) set if \((X,X)\) is an RCT.

\subsection{Optimization of RCTs}
From a conceptual perspective, the main idea of Tube MPC is to optimize an RCT online. Here, the first element of the RCT is enforced to contain the current state measurement. For example, one can solve optimization problems of the form
\begin{align}
\label{eq::TMPC}
\begin{array}{cl}
\displaystyle\min_{X_0,\ldots,X_N} & L_0(X_0) + \displaystyle\sum_{k=0}^{N-1} L(X_k) + L_N(X_N) \vspace{3pt}\\
\text{s.t.} & \left\{
\begin{array}{l}
\forall k \in \{0,\dots,N-1\}, \\
X_{k+1} \in \mathcal{F}(X_k), \\
x \in X_0,\ X_k \subseteq \mathcal X.
\end{array}
\right.
\end{array}
\end{align}
Here, \(L_0\) denotes the initial cost, \(L\) the stage cost, \(L_N\) the terminal cost, and \(x\) the current state measurement. Choices for the cost triple \((L_0, L, L_N)\) that ensure closed-loop stability of such Tube MPC schemes are discussed in~\cite{Villanueva2020}. Alternate formulations may also introduce control-law penalties designed to regularize the control policies, which are hidden in our definition of \(\mathcal F\); for an in-depth discussion, see~\cite{houska2024polyhedralcontroldesigntheory}.

\section{Polytopic Tube MPC Schemes}\label{sec:TubeSchemes}
As the optimization variables of the generic Tube MPC controller~\eqref{eq::TMPC} are set-valued, this problem is in general intractable. Therefore, this section focuses on polytopic set parameterizations.

\subsection{Configuration-Constrained Polytopes}
Throughout this paper, we use the notation
\begin{align}
\label{eq:CCPolytope}
P(y)\coloneqq\{x\in \mathbb{R}^{n_x} \mid Fx \leq y\},
\end{align}
to denote polyhedra with a given facet matrix \(F\in\mathbb{R}^{\fe \times n_x}\). Here, we assume that \(Fx\leq 0\) implies \(x = 0\), so that \(P(y)\) is bounded—and hence potentially empty—for each \(y \in \mathbb{R}^\fe\). In the following, we further refer to the cone
\[
\mathcal{E} := \{y \in \R^{\fe} \mid Ey \leq 0\}
\]
as a configuration cone for \(P\) if the face configuration of the parameterized polytope \(P(y)\) remains invariant for all \(y \in \mathcal{E}\). This means that if \(\mathcal{E}\) is a configuration cone, we can find a set of matrices $V:=\{V_j \mid j \in \{1,\cdots,\ve\}\}$ such that
\begin{align}
\label{eq:vertex_config}
P(y) = \mathrm{convh}(\{V_j y \mid j \in \{1,\cdots,\ve\}\})
\end{align}
for all \(y \in \mathcal{E}\). As already mentioned in the introduction, constructing suitable configuration triples \((F,E,V)\) with the above properties is a non-trivial task. One possible method for constructing such triples can nevertheless be found in~\cite[Sec III.5]{CCTMPC_VILLANUEVA2024111543}. Note that the implications of freezing such a configuration and the methods for refining configuration triples will be discussed below.

\subsection{Configuration-Constrained RCTs}
One of the key advantages of working with configuration-constrained polytopic RCTs, for a given configuration triple \((F,E,V)\), is that they can be represented by a convex condition. Namely, by introducing the convex set
\begin{align}
\label{eq:S_definition}
\mathbb{S}\coloneqq
\left\{ \begin{pmatrix}
    y \\ u \\ y^+
\end{pmatrix} \ \middle| \ 
\begin{aligned}
&\forall \,(i,j) \in \{1,\myldots,\me\} \times \{1,\myldots,\ve\}, \vspace{1pt}\\
& F(A_i V_j y + B_i u_j) + d \leq y^+, \vspace{1pt}\\
& Ey \leq 0, \ V_j y \in \mathcal{X}, \ u_j \in \mathcal{U} 
\end{aligned} 
\right\},
\end{align}
with \(d_k\coloneqq\max\{F_k w : w\in \mathcal{W}\}\) for all \(k \in \{1,\myldots,\fe\}\) and with \(u_j \in \mathcal U\) denoting the \(j\)-th vertex control input, the following statement holds.

\begin{proposition}
\label{prop:RFIT_CC}
Let \(y_0,y_1, \ldots \in \mathcal E\) be a given sequence of parameters. Then \(P(y_0),P(y_1), \ldots\) is a feasible RCT if and only if there exists a sequence of inputs \(u_0,u_1, \ldots \in \R^{\ve \cdot n_u}\) such that \((y_k,u_k,y_{k+1}) \in \mathbb{S}\) for all \(k \in \mathbb{N}\).
\end{proposition}

\textit{Proof:} The statement follows by combining the results from~\cite[Corollary 4]{CCTMPC_VILLANUEVA2024111543} and~\cite[Prop.~1]{Houska2024}.\qed

\begin{corollary}
\label{cor::RCI_CC}
A polytope \(P(y)\) with \(y \in \mathcal E\) is an RCI set if and only if there exists a \(u\in \R^{\ve \cdot n_u}\) such that \((y,u,y) \in \mathbb{S}\).
\end{corollary}

\textit{Proof:} The result follows from Proposition~\ref{prop:RFIT_CC} and the definition of RCI sets.\qed

\subsection{Optimal Configuration-Constrained RCIs}
If \(\ell(\cdot,\cdot)\) is a given convex stage cost, one can use the result from Corollary~\ref{cor::RCI_CC} to compute optimal configuration-constrained RCI polytopes \(P(y_m)\). Here, \(y_m \in \mathcal{E}\) is computed by solving the convex optimization problem
\begin{equation} \label{eq:optimal RCI}
\begin{aligned}
(y_m, u_m) \in \arg\min_{y,u} &\ \ell(y,u)\\
\text{s.t.} &\ (y,u,y)\in\mathbb{S}.
\end{aligned}
\end{equation}
If it is possible to find an explicit expression for the stage cost \(L\) in~\eqref{eq::TMPC} such that \(\ell(y,u) = L(P(y))\), then \(\ell\) does not depend on \(u\). In many practical cases, however, one is interested in refining the objective by introducing a suitable control penalty; thus, \(\ell\) is often modeled as a strongly convex quadratic function, cf.~\cite[Section 2.D]{badalamenti2024CCTMPCTrack}, so that the set of optimizers is a singleton.

\subsection{Fully-parameterized CCTMPC for Tracking}

One way of implementing Tube MPC is to compute a sequence of parameters \(((y_0,u_0),(y_1,u_1),\cdots)\) that converges to the optimal RCI set parameters \((y_m,u_m)\). In the simplest case, this can be achieved by minimizing the quadratic deviation \(\|(y-y_m,u-u_m)\|_Q^2\). For example, one can solve the strictly convex QPs
\begin{equation}
\label{eq:CCTMPC_scheme}
\begin{aligned}
\min_{\mathbf{y},\mathbf{u}}\quad&\sum_{k=0}^{N-1} \norm{ \begin{bmatrix} y_k - y_m \\ u_k-u_m \end{bmatrix}}_Q^2+ \norm{\begin{bmatrix} y_N-y_m \\ u_N-u_m \end{bmatrix}}_R^2
\\
\text{s.t.}\hspace{10pt} &
\left\{
\begin{array}{l}
\forall k \in \{0, \ldots, N - 1 \}, \\
(y_k,u_k,y_{k+1}) \in \mathbb{S}, \\
Fx \leq y_0,\\
(y_N,u_N,\gamma y_N + (1-\gamma)y_m) \in \mathbb{S}\,,
\end{array}
\right.
\end{aligned} 
\end{equation}
online, where \(\mathbf{y}=(y_{0},\myldots,y_N)\) and \(\mathbf{u}=(u_0,\myldots,u_N)\) stack the sequence of tube parameters and vertex control inputs, respectively. Here, \(Q\) and \(R\) are symmetric positive definite matrices in \(\R^{(\fe+\ve\cdot n_u)\times(\fe+\ve\cdot n_u)}\). The parameter \(\gamma \in (0,1)\) is user-specified, and is useful in ensuring recursive feasibility and stability of the control scheme~\cite{badalamenti2024CCTMPCTrack}. Note that Problem~\eqref{eq:CCTMPC_scheme} is parametric in the current state measurement \(x\). We denote the unique parametric minimizers of~\eqref{eq:CCTMPC_scheme} by \((\mathbf{y}^*(x),\mathbf{u}^*(x))\), and define the control law as
\begin{equation}
\label{eq:CCTMPC_control_law}
\begin{aligned}
\mu_{\mathrm{c}}(x):=\arg&\min_{u \in \mathcal{U}} \ u^{\top}u \\ 
&\ \text{s.t.} \
\left\{
\begin{array}{l}
\forall i \in \{ 1, \myldots, \me \},\\
F(A_ix+B_iu)+d \leq y^*_1(x).
\end{array}
\right.
\end{aligned}
\end{equation}

Due to Proposition~\ref{prop:RFIT_CC}, Problem~\eqref{eq:CCTMPC_control_law} is feasible whenever Problem~\eqref{eq:CCTMPC_scheme} is feasible, ensuring that the minimizer \((\mathbf{y}^*(x),\mathbf{u}^*(x))\) exists. The complete closed-loop system is then given by
\begin{align}
\label{eq:CL_CCTMPC}
\forall k \in \mathbb N, \qquad x_{k+1} = A_k x_k + B_k \mu_{\mathrm{c}}(x_k) + w_k.
\end{align}
We recall the following result from~\cite[Corollary 1]{badalamenti2024CCTMPCTrack} regarding the theoretical properties of the controller.
\begin{proposition}
\label{prop:CCTMPC}
Let \(\gamma \in (0,1)\), \(Q\succ0\) and \(R \succ 0\) be given. If \(Q+\gamma^2 R \preceq R\), then the following statements hold independently of the uncertainty realization, \(w_k \in \mathcal{W}\), and \((A_k,B_k) \in \Delta\), for all \(k \in \mathbb{N}\):
\begin{enumerate}
\item If Problem~\eqref{eq:CCTMPC_scheme} is feasible with \(x=x_0\), then it is feasible for \(x=x_k\), where \(x_k\) is the state at time instant \(k\) of System~\eqref{eq:CL_CCTMPC}, for all \(k \in \mathbb{N}\).
\item After defining the function 
$$\mathcal{L}(\mathbf{y},\mathbf{u}):=\sum_{k=0}^{N-1} \norm{ \begin{bmatrix} y_k - y_m \\ u_k-u_m \end{bmatrix}}_Q^2+ \norm{\begin{bmatrix} y_N-y_m \\ u_N-u_m \end{bmatrix}}_R^2,$$
the Lyapunov inequality 
$$\mathcal{L}(\mathbf{y}^*(x_{k+1}),\mathbf{u}^*(x_{k+1})) \leq \mathcal{L}(\mathbf{y}^*(x_{k}),\mathbf{u}^*(x_{k}))$$ 
holds, where equality holds only at \(x_k \in P(y_m)\). Moreover, we have
\begin{align*}
\scalemath{0.95}{\underset{k \to \infty}{\lim} (\mathbf{y}^*(x_k),\mathbf{u}^*(x_k)) = ((y_m,\myldots,y_m),(u_m,\myldots,u_m)).}
\end{align*}
\item The distance between the state in~\eqref{eq:CL_CCTMPC} and the optimal RCI set converges to \(0\); that is, $\underset{k \to \infty}{\lim} x_k \in P(y_m)$.
\end{enumerate}
\end{proposition}
\begin{proof}
The proof follows by applying~\cite[Corollary 1]{badalamenti2024CCTMPCTrack} to the case of fixed optimal RCI set parameters.
\end{proof}

While the TMPC scheme based on~\eqref{eq:CCTMPC_scheme} exhibits favorable theoretical properties as described in Proposition~\ref{prop:CCTMPC}, solving~\eqref{eq:CCTMPC_scheme} online can be computationally demanding. For example, if \(\mathcal{X}\) and \(\mathcal{U}\) are polyhedra represented by \(m_{\mathcal{X}}\) and \(m_{\mathcal{U}}\) half-spaces, respectively, and if \(m_{\mathcal{E}}\) denotes the number of rows of \(E\), then Problem~\eqref{eq:CCTMPC_scheme} has \((N+1) \times (\fe + \ve \cdot n_u)\) optimization variables and \((N+1) \times (\fe\me\ve + m_{\mathcal{E}} + \ve(m_{\mathcal{X}} + m_{\mathcal{U}}))\) constraints. Here, we additionally recall that \(\fe\) denotes the number of facets, \(\ve\) the number of vertices, and \(n_u\) the number of control inputs.

\subsection{Homothetic CCTMPC for Tracking}
One path to reducing the number of optimization variables in the above Tube MPC formulation is to work with a homothetic transformation. To outline the main idea, the following standard assumption is introduced (see also~\cite{HTMPC_RAKOVIC20121631}).
\begin{assumption}
\label{ass:C_mRPI}
The optimal RCI set satisfies the condition $\scalemath{0.95}{(0,0) \in P(y_m) \times \mathrm{convh}(\{u_{m,j} \in \mathcal{U} \mid j \in \{1,\myldots,\ve\}\})}$.
\end{assumption}
Next, similar to~\cite{HTMPC_RAKOVIC20121631}, we introduce a scaling factor $\alpha \in \R$ and offset vectors $(z,v) \in \R^{n_x+n_u}$, such that
\begin{align}
    \label{eq:HTMPC_parameterization}
    y=\alpha y_m+Fz, \quad u_j=\alpha u_{m,j}+v, \ j \in \{1,\myldots,\ve\}.
\end{align}
We recall from~\cite{CCTMPC_VILLANUEVA2024111543} that the construction of the matrix \(E\) and the associated configuration cone \(\mathcal{E}\) imply
\begin{align}
\label{eq:HTMPC_CC}
y_m \in \mathcal{E} \ \Rightarrow \ \alpha y_m+Fz \in \mathcal{E}, \quad \forall \alpha\geq 0,\ z \in \R^{n_x}.
\end{align}
Thus, \(P(\alpha y_m+Fz)\) is a configuration-constrained polytope for all \(\alpha \geq 0\). To enforce the state and input constraints, we assume that the constraint sets are modeled by polytopes, namely, $\mathcal{X}:=\{x \in \R^{n_x} \mid H^x x \leq h^x \}$ and $\mathcal{U}:=\{u \in \R^{n_u} \mid H^u u \leq h^u\}$. Next, we define 
\begin{align*}
    h^x_m&:=\max\{H^xV_j y_m \mid j \in \{1,\myldots,\ve\}\}, \\
    h^u_m&:=\max\{H^u u_{m,j} \mid j \in \{1,\myldots,\ve\}\},
\end{align*}
with the $\max$ taken row-wise. It follows that the inclusions $V_j(\alpha y_m+Fz) \in \mathcal{X}$ and $\alpha u_{m,j}+v \in \mathcal{U}$ hold for all vertices $j \in \{1,\cdots,\ve\}$ if and only if
\begin{align}
    \label{eq:HTMPC_constraints}
    \begin{bmatrix} H^x & 0 \\ 0 & H^u \end{bmatrix} \begin{bmatrix} z \\ v \end{bmatrix} + \begin{bmatrix} h^x_m \\ h^u_m \end{bmatrix} \alpha \leq \begin{bmatrix} h^x \\ h^u \end{bmatrix}. 
\end{align}
Finally, we define the set $\bar{\mathbb{S}} \subseteq \R^{2n_x+n_u+2}$ by
\begin{align*}
\mathbb{\bar{S}}\coloneqq
\left\{ \begin{pmatrix}
    (z,\alpha) \\ v \\ (z^+\!,\alpha^+)
\end{pmatrix}  \middle| \, 
\begin{aligned}
&\alpha \geq 0,\quad \eqref{eq:HTMPC_constraints}, \quad  \forall \,i \in \{1,\myldots,\me\}, \vspace{1pt}\\
& \begin{multlined}
    F(A_i z + B_i v) + (1\!-\!\alpha) d + \alpha y_m \\
\leq \alpha^+y_m + Fz^+,
\end{multlined}
\end{aligned} 
\right\}\!.
\end{align*}

\begin{corollary}
\label{corr:RFIT_HTMPC}
Let \((z_k,\alpha_k)_{k \in \mathbb{N}}\) be such that for all \(k \in \mathbb{N}\) there exist inputs \(v_k\) satisfying $((z_k,\alpha_k),v_k,(z_{k+1},\alpha_{k+1})) \in \bar{\mathbb{S}}$. Then, the sequence $(P(\alpha_k y_m+Fz_k))_{k \in \mathbb{N}}$ is an RCT.
\end{corollary}
\begin{proof}
The state, input, and configuration constraints for the RCT follow from~\eqref{eq:HTMPC_constraints} and~\eqref{eq:HTMPC_CC}. Substituting the parameterization in~\eqref{eq:HTMPC_parameterization} into the set \(\mathbb{S}\) and applying Proposition~\ref{prop:RFIT_CC}, along with the observation that \((y_m,u_m)\) satisfy $F(A_i V_j y_m + B_i u_{m,j})+d \leq y_m$ for all $(i,j) \in \{1,\myldots,\me\} \times \{1,\myldots,\ve\}$, completes the proof.
\end{proof}

By exploiting Corollary~\ref{corr:RFIT_HTMPC}, we formulate our Homothetic TMPC scheme based on the parametric QP
\begin{equation}
\label{eq:HTMPC_scheme}
\begin{aligned}
\min_{\mathbf{z},\boldsymbol{\alpha},\mathbf{v}} &\ \sum_{k=0}^{N-1}\norm{ \begin{bmatrix} z_k \\ v_k \\ \alpha_k-1 \end{bmatrix}}_{\bar{Q}}^2+ \norm{\begin{bmatrix} z_N \\ v_N \\ \alpha_N-1 \end{bmatrix}}_{\bar{R}}^2 \vspace{3pt} \\
\text{s.t.} \hspace{4pt} & \left\{\!
\begin{aligned}
&\forall k \in \{ 0, \ldots,N - 1\},\\
&((z_k,\alpha_k),v_k,(z_{k+1},\alpha_{k+1})) \in \bar{\mathbb{S}}\\
&Fx \leq \alpha_0 y_m + Fz_0,\\
&((z_N,\alpha_N),v_N,(\gamma z_N,\gamma\alpha_N+1-\gamma)) \in \bar{\mathbb{S}},
\end{aligned}
\right.
\end{aligned}
\end{equation}
where $\mathbf{z}=(z_{0},\myldots,z_N)$, $\boldsymbol{\alpha}=(\alpha_{0},\myldots,\alpha_N)$, and $\mathbf{v}=(v_0,\myldots,v_N)$ stack the sequence of homothetic tube parameters, and $\bar{Q}$ and $\bar{R}$ are symmetric positive definite matrices in $\R^{(n_x+n_u+1)\times(n_x+n_u+1)}$. We denote the optimizers of~\eqref{eq:HTMPC_scheme} as $(\mathbf{z}^*(x),\boldsymbol{\alpha}^*(x),\mathbf{v}^*(x))$, and define the control law as
\begin{align}
\label{eq:HTMPC_control_law}
\mu_{\mathrm{h}}(x):=\arg\min_{u \in \mathcal{U}}& \ u^{\top}u \\ 
\text{s.t.} \hspace{2pt}&  
\left\{\,
\begin{aligned}
&\forall i \in \{1,\myldots,\me\},\\
&F(A_ix+B_iu)+d \leq \alpha^*_1 y_m+Fz^*_1(x).
\end{aligned}
\right. \nonumber
\end{align}

Finally, we write the closed-loop system as
\begin{align}
\label{eq:CL_HTMPC}
x_{k+1} = A_k x_k + B_k \mu_{\mathrm{h}}(x_k) +w_k.
\end{align}
\begin{corollary}
Suppose that \(\gamma \in (0,1)\) and the cost matrices \(\bar{Q}\succ0\) and \(\bar{R} \succ 0\) satisfy \(\bar{Q}+\gamma^2 \bar{R} \preceq \bar{R}\). Then, the system in~\eqref{eq:CL_HTMPC} enjoys properties equivalent to those in Proposition~\ref{prop:CCTMPC}.
\end{corollary}
\begin{proof}
Define the matrix \(\tilde{\I} = \mathbf{1}_{\ve} \otimes \I_{n_u}\) and 
\begin{align}
\label{eq:trackCostEquiv}
\tilde{F}:=\begin{bmatrix} F & 0 & y_m \\ 0 & \tilde{\I} & u_m \end{bmatrix} \in \R^{(\fe+\ve \cdot n_u) \times (n_x+n_u+1)}.
\end{align}
We observe that Problem~\eqref{eq:HTMPC_scheme} is equivalent to Problem~\eqref{eq:CCTMPC_scheme} with the additional optimization variables \((\mathbf{z},\boldsymbol{\alpha},\mathbf{v})\), linear equality constraints \(y_k = \alpha_k y_m + Fz_k\) and \(u_k = \alpha_k u_m + \tilde{\I}\,v_k\), and with cost matrices given by  \(Q=\|\tilde{F}^\dagger\|^2_{\bar{Q}}\) and \(R=\|\tilde{F}^\dagger\|^2_{\bar{R}}\), where \(\tilde{F}^{\dagger}\) is well-defined since \(\tilde{F}\) is full column-rank. Recursive feasibility follows by exploiting the convexity of the set \(\bar{\mathbb{S}}\) and the linear equality constraints, and Lyapunov stability follows from the fact that \(Q\) and \(R\) are positive definite on the nullspace of the equality constraints, together with the inequality \(\bar{Q} + \gamma^2\, \bar{R} \preceq \bar{R}\) implying \(Q + \gamma^2 R \preceq R\).
\end{proof}

We remark that Problem~\eqref{eq:HTMPC_scheme} is computationally less demanding to solve online than Problem~\eqref{eq:CCTMPC_scheme}, since it has \((N+1) \times (n_x+n_u+1)\) optimization variables and \((N+1) \times (\fe\me + 1 + m_{\mathcal{X}} + m_{\mathcal{U}})\) constraints.

\subsection{A family of CCTMPC schemes}
The schemes presented in~\eqref{eq:CCTMPC_scheme} and~\eqref{eq:HTMPC_scheme} represent the two extremes of leveraging configuration-constrained polytopes for tube-based model predictive control (TMPC) synthesis. While~\eqref{eq:HTMPC_scheme} is based on a restriction of the parameters \((y,u)\) in terms of \((y_m,u_m)\) via~\eqref{eq:HTMPC_parameterization}, alternative formulations of CCTMPC can be derived by adopting more flexible parameterization strategies. These intermediate approaches aim to achieve a superior trade-off between computational efficiency and control conservatism for practical applications.

To demonstrate this flexibility, we propose two such intermediate schemes. For both formulations, admissible control sets equivalent to \(\mathbb{S}\) are systematically constructed by substituting the proposed \((y,u)\)-parameterizations into the synthesis framework, while explicitly leveraging the inherent inclusion \((y_m,u_m,y_m) \in \mathbb{S}\) to ensure recursive feasibility.

\subsubsection{Partially-parameterized CCTMPC}
In this scheme, we parameterize some vertices of \(P(y)\) as homothetic transformations of the corresponding vertices of \(P(y_m)\), while allowing the remaining vertices to be freely optimized. Let \(J_{\mathrm{p}} \subseteq \{1,\myldots,\ve\}\) be the set of indices corresponding to the vertices of \(P(y)\) that we want to parameterize as
\begin{align}
\label{eq:vertices_scaled}
    V_j y = \alpha V_j y_m+z, \quad \forall j \in J_{\mathrm{p}}.
\end{align}
For each \(j\), denote the indices of the halfspaces intersecting at the vertex \(V_j y\) by \(i(j) \subset \{1,\myldots,\fe\}\). To achieve~\eqref{eq:vertices_scaled}, we define the index set $I_{\mathrm{p}}:= \cup_{j \in J_{\mathrm{p}}} i(j)$. Additionally, we use the symbol \(e_i\) to denote the \(i\)-th column of the identity matrix \(\mathbb{I}_{\fe}\), thus introducing the matrix
\begin{align*}
    S_{\mathrm{p}} = \sum_{i \in I_{\mathrm{p}}} e_i e_i^{\top}.
\end{align*}

Furthermore, we define the set $I_{\mathrm{f}}:=\{1,\myldots,\fe\}\setminus I_{\mathrm{p}}$. Then, the polytope $P(\alpha S_{\mathrm{p}} y_m+Fz+S_{\mathrm{f}} \hat{y})$ satisfies~\eqref{eq:vertices_scaled} if $\alpha S_{\mathrm{p}} y_m+Fz+S_{\mathrm{f}} \hat{y} \in \mathcal{E}$. For consistency with~\eqref{eq:vertices_scaled}, we parameterize the vertex inputs for $j \in J_{\mathrm{p}}$ as $ u_j = u_{m,j}+v$, while the inputs corresponding to the remaining vertices are freely optimized. 

Denoting the cardinality of \(J_{\mathrm{p}}\) and \(I_{\mathrm{f}}\) as \(\hat{\ve}\) and \(\hat{\fe}\) respectively, the Partially-parameterized CCTMPC scheme is then composed of $(N+1) \times (\hat{\fe}+n_x+n_u+1+(\ve-\hat{\ve})n_u)$ optimization variables, with \(\hat{\ve}>0\). Note that if \(\hat{\ve}=\ve\), then by construction \(\hat{\fe}=0\) and the Partially-parameterized CCTMPC scheme reduces to the Homothetic CCTMPC scheme. Conversely, for \(\hat{\ve}=0\), the scheme retrieves the Fully-parameterized CCTMPC formulation. The number of linear inequality constraints is the same as that of the Fully-parameterized CCTMPC scheme for any \(\hat{\ve} \in [0,\ve)\), while it reduces to that of Homothetic CCTMPC if \(\hat{\ve} = \ve\). In practice, the index set \(J_{\mathrm{p}}\) should be chosen so that the controller’s stabilizable region includes the desired system’s initial states. Developing a systematic method for selecting \(J_{\mathrm{p}}\) based on a given set of initial conditions is a topic for future research.

\subsubsection{Homothetic CCTMPC with optimized controls}
Inspired by the approach of~\cite{Langson2004}, we parameterize \(y=\alpha y_m+Fz\), while allowing the vertex control inputs to be freely optimized. This results in a scheme defined by \((N+1) \times (n_x+1+\ve \cdot n_u)\) optimization variables and \((N+1) \times (\fe\me\ve+1+m_{\mathcal{X}}+\ve m_{\mathcal{U}})\) linear inequality constraints, with a stabilizable region that lies between the Homothetic CCTMPC and the Fully-parameterized CCTMPC schemes.

\section{Computing Template Polytopes} \label{sec:TemplateCompute}
This section presents a computational procedure for the construction of configuration triples that can be used in TMPC synthesis. In particular, the focus is on iteratively refining a configuration triple \((F,E,V)\) such that, at each iteration, \(\mathbb{S}\) (or \(\bar{\mathbb{S}}\)) remains nonempty, the number of vertices of the new polytope increases by a known, fixed amount, and the size of the resulting RCI set increases. 

\subsection{Iterative Refinement of the Configuration Triple}

Suppose a configuration triple $(F,E,V)$ is given such that $\mathbb S$ is nonempty. Let $\mathcal X = \operatorname{convh}(\{\xi_1,\myldots,\xi_{\rm s}\})$, and consider the optimization 
problem
\begin{align}
\label{eq:size_function}
    \sigma \coloneqq \min_{y_{M}\geq0,u,z} &\ \sum_{i=1}^{v_{\mathcal{X}}}\|\xi_i-z_i\|_2^2 \\
    \text{s.t.} \hspace{8pt}&\ (y_{M},u,y_{M}) \in \mathbb{S},\  F z_i \leq y_{M},\ i \in \{1,\myldots,{\rm s}\}\,.\nonumber
\end{align}

Notice that, although we do not report the dependencies, the above problem is parametric in the configuration triple. Let \((y_{M}^\star,u^\star,z^\star)\) be a minimizer of~\eqref{eq:size_function}. Then, the RCI polytope $\{ x \in \mathcal X \ | \  F x \leq y_{M}^\star \}$ contains points $z^\star_{1},\ldots,z^\star_{\rm s}$ that have minimal aggregated distance to the vertices of $\mathcal X$. 
\begin{algorithm}
\caption{
Configuration Triple Refinement
}\label{alg:refinement}
\KwData{Initial triple $(F,E,V)$, $i_{\mathrm{max}} \geq 1$}
$i \gets 1$\;
Solve~\eqref{eq:size_function} with $(F,E,V)$, set  $\sigma^i\gets \sigma$ and $y^1 \leftarrow y_M^\star$\;
    \While{$i \leq i_{\mathrm{max}}$}{
         Set $c_j \gets V_j\, y^i$ for $j \in \{1,\myldots,\ve\}$\;
        \textbf{par}\For{$j \gets 1$ \KwTo $\ve$}{
            $\zeta \gets \max_{z} \, c_j^{\top} z \ \text{s.t.} \ Fz\leq y^i$\;
            Set $\kappa \in (0,1)$ such that 
            $c_j^{\top} c_k < \kappa \zeta,\quad \forall k \in \{1,\myldots,\ve\}\setminus\{j\}$\;
             Compute the configuration triple $(F^j\!,E^j\!,V^j)$ for the polytope $\{ x \, | \, Fx \leq y^i, \, c_j x \leq \kappa\zeta \}$\;
             Solve~\eqref{eq:size_function} with $(F^j,E^j,V^j)$, set $\overline{\sigma}^j  \gets \sigma$ and $\overline{y}^j\gets y^\star_{M}$\;
        }
        Compute $j^\star$ such that $\overline{\sigma}^{j^\star} \leq \overline{\sigma}^j$  $\forall j\in\{1,\myldots,\ve\}$\;
        Set $(F,E,V) \gets (F^{j^\star}\!,E^{j^\star}\!,V^{j^\star})$\;
        Set $i \gets i+1$, $\sigma^{i}\gets \overline{\sigma}^{j^\star}$, and $y^i \gets \overline{y}^{j^\star}$\;
    }
    \Return $(F,E,V)$.
\end{algorithm}

Algorithm~\ref{alg:refinement} presents a procedure that iteratively adds rows to the matrix \(F\) while ensuring that \(\sigma\) is nonincreasing as the configuration triple is refined. Note that the main steps are contained in the inner (parallelizable) for loop in Steps 6 to 9. First, for each vertex \(c_j\), a supporting hyperplane \(\{ x \mid c_j x \leq \zeta \}\) is computed (Step 6). This hyperplane is then rescaled by \(\kappa\) so that \(c_j\) is excluded from the polytope \(\{ x \mid Fx \leq y^i,\; c_j x \leq \kappa\zeta \}\) (Step 7). At this point, a new template for this polytope is computed (Step 8), and the one that produces the largest RCI set is kept (Steps 9 to 11). Note that in a practical implementation of Algorithm~\ref{alg:refinement}, Step 4 must be modified to handle vertex redundancies: if there exists \(c_i = c_j\) for some \(i \neq j\), then one of them can be removed from the vertex list so that the enumeration in the inner for loop is performed over unique vertices.
\begin{proposition}
\label{prop:alg}
The sequence $(\sigma^1,\myldots,\sigma^{i_{\mathrm{max}}})$ generated by Algorithm~\ref{alg:refinement} is nonincreasing.
\end{proposition}
\begin{proof}
    Let \(i\) be an arbitrary iteration index and \((y_M,u,z)\) a minimizing triple of~\eqref{eq:size_function} with optimal value \(\sigma^i\). Then, for an arbitrary vertex index, consider Step 9 of the algorithm. Here, we have that \(z\) is feasible together with \(\bigl((y_M^\top, \zeta)^\top\bigr)\) and \(u\) is chosen so that the vertex inputs corresponding to the new vertices coincide with the vertex input \(u_j\) (recall that the halfspace \(\{ x \mid c_j x \leq \zeta \}\) is redundant per Step 6). Thus, \(\overline{\sigma}^{j} \leq \sigma^i\). The proof concludes since, by Steps 10 and 12, \(\sigma^{i+1} \leq \overline{\sigma}^{j}\) for all \(j \in \{1,\myldots,\ve\}\).
\end{proof}
\begin{remark}
    A discussion on the construction of configuration triples can be found in~\cite[Section~III.5]{CCTMPC_VILLANUEVA2024111543} and~\cite[Remark 1]{Houska2024}. A more thorough discussion is available in~\cite[Sec II.13]{houska2024polyhedralcontroldesigntheory}, where it is explained that, in the case of triples constructed from simple polytopes, the number of rows of \(E\) is bounded by the number of edges of the polytope. Such observations can be used to improve the efficiency of Algorithm~\ref{alg:refinement}.
\end{remark}

\subsection{Computing Initial Matrix \(F\)}
We present a nonlinear program (NLP) inspired by~\cite{mulagaleti2024PDRCI,Gupta2020} to compute an initial feasible configuration triple for Algorithm~\ref{alg:refinement}. In particular, given a polytope
\[
\{ x\in\mathbb{R}^{n_x} \mid Fx \leq 1  \} = 
\operatorname{convh}(\{V_11,\myldots,V_{\ve}1\}),
\]
the procedure computes an invertible matrix \(T\in\mathbb{R}^{n_x \times n_x}\) such that $\{ x\in\mathbb{R}^{n_x} \mid FT^{-1}x \leq 1  \}$
is an RCI polytope with vertices $\{TV_11,\ldots,TV_{\ve}1\}$.

This NLP is formulated as
\begin{equation}
\label{eq:NLP_template}
    \begin{aligned}
        \min_{T,u,\varepsilon} &\quad \varepsilon^{\top}\varepsilon \\
    \text{s.t.} \hspace{3pt} &\left\{
    \begin{aligned}
    &\forall (i,j,k) \in \{1,\myldots,\me\}\! \times\! \{1,\myldots,\ve\}\! \times \! \{1,\myldots,\mathsf{w}\},\\    
    &FT^{-1}(A_i T z_j+B_iu_j + w_k) \leq 1,\\
    &Tz_j \in \mathcal{X}, \ u_j \in \mathcal{U},\\
    &\mathcal{X} \subseteq \{x \mid FT^{-1}x \leq 1+\varepsilon\}. 
    \end{aligned}
    \right.
    \end{aligned}
\end{equation}
with $\{w_1,\ldots,w_\mathsf{w}\}$ denoting the vertices of $\mathcal W$.

The inclusion $\mathcal{X} \subseteq \{x \mid FT^{-1}x \leq 1+\varepsilon\}$ in~\eqref{eq:NLP_template} can be formulated, using duality, as
\begin{align*}
\Lambda h^x \leq 1+\varepsilon,\ \Lambda H^x = FT^{-1},\ \Lambda\geq 0,
\end{align*}
with $\Lambda\in\R^{\fe\times \mathsf{w}}$ being an additional optimization variable.

Notice that the face configuration (and thus also the cone \(\mathcal{E}\)) of the original polytope is preserved since \(T\) defines an invertible linear transformation (see, e.g.,~\cite[Section 2.6]{Ziegler2008-hp} and~\cite[Section 1.1]{Grunbaum2003-xh}).

\section{Numerical Examples} \label{sec:NumExamples}
We validate our procedural construction using two numerical examples\footnote{MATLAB code available at~\url{https://github.com/fil-bad/EfficientCCTMPC}}. In both cases, we choose the set $\mathcal{Z}$ as an $n_x$-simplex with 
$F= [-\mathbb{I}_{n_x}\ \boldsymbol{1}]^\top$, ensuring that the origin is always in its interior for any positive offset vector. For the terminal set, \(\gamma\) is set to \(0.95\). We denote the stabilizable region of the Fully-parameterized CCTMPC scheme with horizon length \(N\) as \(\mathcal{O}_{\mathrm{c}}(N)\), and that of the Homothetic CCTMPC (HTMPC) scheme as \(\mathcal{O}_{\mathrm{h}}(N)\). Computational times for Problem~\eqref{eq:CCTMPC_scheme} and~\eqref{eq:HTMPC_scheme} were collected on an Intel Core i5-8350U CPU running Ubuntu 24.04.

\subsection{Illustrative Example} \label{subsec:3DIntegrator}
We consider a discretized triple integrator subject to both additive and multiplicative uncertainties. The nominal state and input matrices are defined as
\begin{align*}
    \bar{A} = \begin{bmatrix}
    1 & h & \sfrac{h^2}{2} \\
    0 & 1 & h \\
    0 & 0 & 1 
    \end{bmatrix}, 
    \quad \bar{B} = \begin{bmatrix}
        \sfrac{h^3}{6}\\
        \sfrac{h^2}{2}\\
        h
    \end{bmatrix}
\end{align*}
with $h=\sfrac{1}{4}$. The uncertain convex hull is given by $A =(1\pm0.1)\bar{A}$, $B =(1\pm0.1)\bar{B}$. The additive disturbance set is
\begin{align*}
    \mathcal{W} = \left\{ \begin{bmatrix}
    h & \sfrac{h^2}{2} & \sfrac{h^3}{6} \\
    1 & h & \sfrac{h^2}{2} \\
    0 & 1 & h 
    \end{bmatrix}w \ \middle|\ w\in \frac{1}{20}[-1,\, 1]^3 \right\},
\end{align*}
and the constraints are $\mathcal{X}=[-5,\,5]^3$ and $\mathcal{U}=[-3,\,3]$.
The initial template obtained from Problem~\eqref{eq:NLP_template} is
\begin{align*} FT^{-1} = 
\begin{bmatrix}
    1.1856&    2.1991&    0.2544\\
    0&    1.4770&    1.7581\\
   -2.6514&   -5.3810&   -2.6623\\
    1.4658&    1.7048&    0.6498\\
\end{bmatrix}.
\end{align*}
Algorithm~\ref{alg:refinement} is executed for \(i_{\max}=20\) iterations, where we search over the class of polytopes with $\fe= n_x+1+i$ facets and $\ve=n_x+1 + i(n_x-1)$ vertices, for \(i\in[1,i_{\max}]\). Figure~\ref{fig:algorithm1} illustrates the iterative refinement process of the configuration-constrained template.
\begin{figure}
    \centering
    \includegraphics[width=1\linewidth]{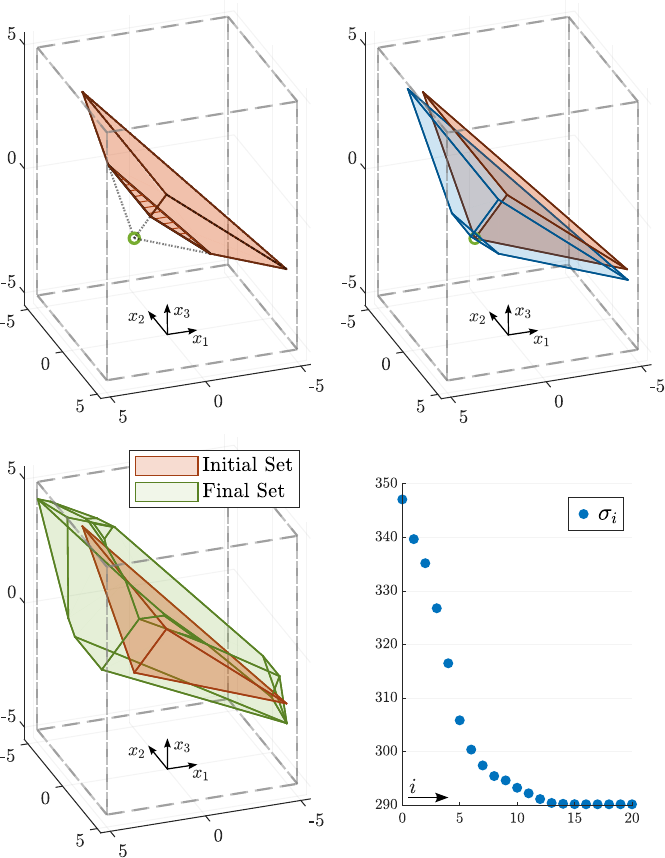}
    \caption{Steps of Algorithm~\ref{alg:refinement}: (top-left) First iteration—selected vertex $i^*$ (green) is marked for cutting, with the shadowed facet indicating the cut at $\kappa\zeta$. (top-right) CC-RCI set (blue) after one iteration. \texttt{Bottom-Left:} CC-RCI set before (red) and after $i_{\max}$ iterations (green). \texttt{Bottom-Right:} Evolution of $\sigma^i$ versus the iteration index~$i$.}
    \label{fig:algorithm1}
\end{figure}

To implement our control schemes, we select the template from iteration $i=10$. The optimal RCI cost in~\eqref{eq:optimal RCI} is defined as
\begin{equation*}
\ell(y,u) = \sum^{\ve}_{j=1} \norm{\begin{bmatrix}
    (\bar{V} -V_j)y \\ (\bar{U} -U_j)u 
\end{bmatrix}}^{2}_{Q_{v}}\! + 
\norm{\begin{bmatrix}
\bar{V} & 0 \\
0 & \bar{U}
\end{bmatrix}
\begin{bmatrix}
y \\ u
\end{bmatrix}}_{Q_{c}}^2,
\end{equation*}
with $\bar{V}=\sum_{j=1}^\ve V_j$, $\bar{U}=\sum_{j=1}^\ve U_j$, $Q_v = 10^{-1}\I_4$, and $Q_c = \I_4$. The tracking cost matrices are set to $Q = \I_{38}$ and $R=(1-\gamma^2)^{-1}Q$ for CCTMPC, with the proper adaptation as in~\eqref{eq:trackCostEquiv} for HTMPC. The prediction horizon is chosen as $N=3$.

Figure~\ref{fig:3D_FeasRegions} compares the sets \(\mathcal{O}_{\mathrm{c}}(3)\) and \(\mathcal{O}_{\mathrm{h}}(3)\), and illustrates several sampled trajectories converging to the optimal RCI set. Denoting \(\mathcal{X}^{-}\) as the 6-step backward reachable set from \(\mathcal{X}\), the Hausdorff distance metrics are $\mathrm{d}_{\mathcal{X}^{-}}(\mathcal{O}_{\mathrm{c}}(3))=3.3925$ and $\mathrm{d}_{\mathcal{X}^{-}}(\mathcal{O}_{\mathrm{h}}(3))=5.8220$ respectively, showing that CCTMPC is less conservative than HTMPC. While these metrics should ideally be computed against the Maximal RCI set, we report that the backward iteration procedure to compute the set runs out of memory after 6 iterations. Finally, by using the DAQP solver~\cite{arnstrom2022dual}, the (\textit{minimum}, \textit{average}, \textit{maximum}) computational times over sampled trajectories are \((2.59, 15.06, 76.96)\)~[ms] for CCTMPC and \((1.04, 1.30, 3.67)\)~[ms] for HTMPC, highlighting the trade-off between increased flexibility and computational cost—a discrepancy that grows with template complexity.
\begin{figure}
    \centering
    \includegraphics[width=1\linewidth]{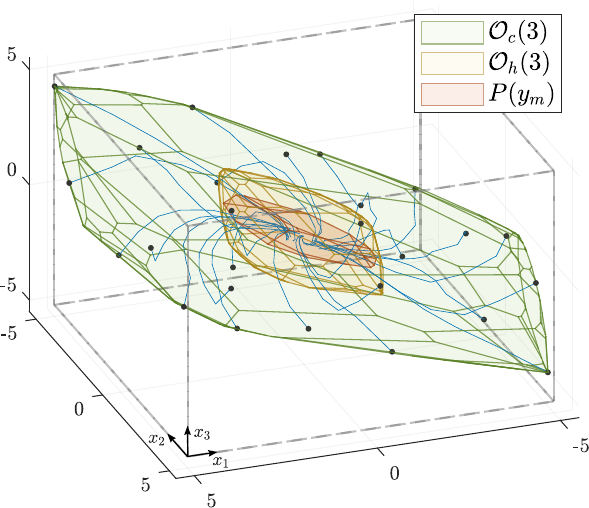}
    \caption{Comparison of feasible regions. State constraints (gray) and sampled trajectories (blue) are included.}
    \label{fig:3D_FeasRegions}
\end{figure}

\subsection{Quadrotor System} \label{subsec:Quadrotor}

We consider a $10$-state dynamical model of a quadrotor as in~\cite{TubeMPCQuadrotor}, discretized with a timestep of $0.1\,\mathrm{s}$ and linearized at $x_{eq}=0$ and $u_{eq} = [0, 0, k_T^{-1}g]^\top$, with $k_T = 0.91$ and $g$ denoting the gravitational constant. The state constraints are given by $\mathcal{X}=\{x\in \pm[4,4,2,10,10,5,\sfrac{\pi}{3},\sfrac{\pi}{3},\pi,\pi]\}$, while the inputs constraints are $\mathcal{U} = \{(u_1,u_2)\in \pm[\sfrac{\pi}{4}]^2, u_3\in[0,2g]-(k_T^{-1}g)\}$. The additive disturbance set is given by $\mathcal{W} = \{B_w w \mid w\in\pm[0.05,0.05,0.1]\}$, with $B_w = [\I_3\ \boldsymbol{0}_{3\times7}]^\top$.

For this high-dimensional example, we first illustrate the benefits of using the homothetic parameterization over the fully-parameterized CCTMPC scheme for synthesizing controllers. To this end, we compute a template matrix \(F\) by solving~\eqref{eq:NLP_template}, which we then enrich by introducing a cutting plane at each of the vertices. This results in a polytope with \(\fe = 22\) facets and \(\ve = 110\) vertices. Although Algorithm~\ref{alg:refinement} could be used to enrich the template systematically, we report that \(\sigma^i = \sigma^1\) for all \(i>1\), indicating that the algorithm is blocked at a local minimum. Strategies to escape local minima (e.g., picking an arbitrary cutting plane instead of following Step 10) are currently under study.

Optimal RCI cost is computed using $\ell(y,u)=\norm{y}_2^2+\norm{u}_2^2$, while the tracking cost is given by \(\bar{Q}=\I_{14}\), \(\bar{R}=(1-\gamma^2)^{-1}\bar{Q}\) for HTMPC and \(Q=\I_{352}\), \(R=(1-\gamma^2)^{-1}Q\) for CCTMPC. In Figure~\ref{fig:HausTimeComparison}, we compare the Hausdorff distance metric associated to the HTMPC scheme, \(\mathrm{d}_{\mathcal{X}}(\mathcal{O}_{\mathrm{h}}(N))\) for different $N$, against that of CCTMPC formulated with a prediction horizon \(N=3\). For any \(N\geq5\), we obtain $\mathrm{d}_{\mathcal{X}}(\mathcal{O}_{\mathrm{h}}(N))<\mathrm{d}_{\mathcal{X}}(\mathcal{O}_{\mathrm{c}}(3))$. In the same figure, the average computational time for CCTMPC$(3)$ (334~[ms]) is compared against HTMPC$(N)$ by using the Gurobi solver~\cite{gurobi}, as the dimensionality of this optimization problem makes Interior Point solvers attractive. The runtime advantage of the homothetic scheme is evident, as it is irrespective of the chosen template complexity.
\begin{figure}
    \centering
    \includegraphics[width=1\linewidth]{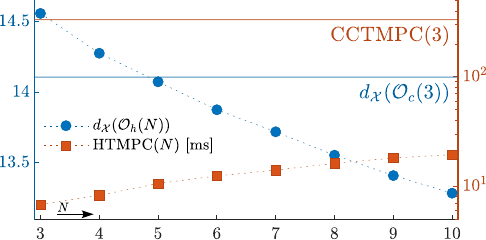}
    \caption{Performance comparison of HTMPC for different prediction horizons \(N\) against a CCTMPC\((3)\) baseline: Hausdorff distance to state constraints (blue) and average computational time (red).}
    \label{fig:HausTimeComparison}
\end{figure}

Motivated by these observations, we design a TMPC controller based on the HTMPC scheme. We reduce conservativeness by using a template matrix \(FT^{-1}\) obtained by solving Problem~\eqref{eq:NLP_template} with $F=[\I_{10} -\I_{10}]^{\top}$, resulting in $\ve=1024$. The closed-loop spatial trajectories obtained from Problem~\eqref{eq:HTMPC_scheme} formulated with $N=20$ are shown in Figure~\ref{fig:10StateQuadrotor}, where we also observe that the tube $P(\alpha_0y_m+Fz_0)$ converges to the optimal RCI set.
Using this template, we report that $\mathrm{d}_{\mathcal{X}}(\mathcal{O}_{\mathrm{h}}(20))=10.355$, and Problem~\eqref{eq:HTMPC_scheme} requires an average of 22.9~[ms] per iteration.
\begin{figure}
\centering
\includegraphics[width=\linewidth]{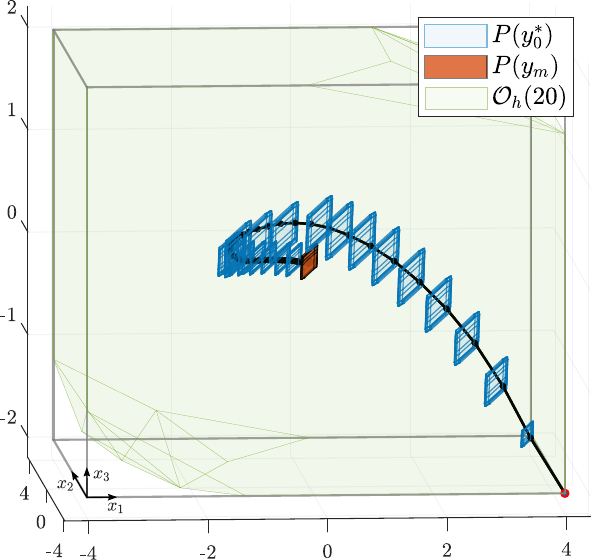}
\caption{Closed-loop evolution of HTMPC projected onto the spatial coordinate system. The red dot indicates the initial state.}
\label{fig:10StateQuadrotor}
\end{figure}

\section{Conclusions} \label{sec:Conclusions}
This paper advances the recently introduced CCTMPC framework for synthesizing TMPC schemes by $(i)$ presenting a family of RCT parameterizations and corresponding control schemes that strategically trade off between computational complexity and conservativeness, and $(ii)$ developing an algorithm to systematically synthesize template polytopes of user-specified complexity for CCTMPC synthesis. The benefits of these advances were demonstrated through numerical examples, paving the way for applying CCTMPC schemes to the design of robust controllers for relatively high-dimensional systems such as the \(10\)-state quadrotor benchmark illustrated in this article. Future work will focus on further refining Algorithm~\ref{alg:refinement} and developing an efficient software tool to synthesize template polytopes that exploit the properties of configuration constraints.

\balance
\bibliographystyle{plain}
\bibliography{references}

\end{document}